\documentclass[12pt,draftcls,peerreview, onecolumn]{IEEEtran}

\usepackage{times, epsf}
\usepackage{amsmath,amssymb}
\usepackage{graphicx}
\usepackage{color}
\usepackage{subfigure}
\usepackage{cite}
\usepackage{multirow,tabularx}

\usepackage{wasysym}
\usepackage{stackrel}
\usepackage{ifthen}

\newcommand{\hide}[1]{\ifthenelse{\boolean{false}}{#1}{}}

\newtheorem{theorem}{{\bf Theorem}}
\newtheorem{lemma}{{\bf Lemma}}
\newtheorem{proposition}{Proposition}
\newtheorem{corollary}{{\bf Corollary}}

\newcommand{\qed}{\nobreak \ifvmode \relax \else
      \ifdim\lastskip<1.5em \hskip-\lastskip
      \hskip1.5em plus0em minus0.5em \fi \nobreak
      \vrule height0.75em width0.5em depth0.25em\fi}

\newcommand{\beq}{\begin{equation}}
\newcommand{\eeq}{\end{equation}}

\newcommand{\bdisp}{\begin{displaymath}}
\newcommand{\edisp}{\end{displaymath}}

\newcommand{\beqarr}{\begin{eqnarray}}
\newcommand{\eeqarr}{\end{eqnarray}}

\newcommand{\bmlt}{\begin{multline}}
\newcommand{\emlt}{\end{multline}}

\newcommand{\beqarrn}{\begin{eqnarray*}}
\newcommand{\eeqarrn}{\end{eqnarray*}}

\newcommand{\barr}{\begin{array}}
\newcommand{\earr}{\end{array}}

\newcommand{\benum}{\begin{enumerate}}
\newcommand{\eenum}{\end{enumerate}}

\newcommand{\bit}{\begin{itemize}}
\newcommand{\eit}{\end{itemize}}

\newcommand{\bc}{\begin{center}}
\newcommand{\ec}{\end{center}}

\newcommand{\bdes}{\begin{description}}
\newcommand{\edes}{\end{description}}

\newcommand{\bfig}{\begin{figure}}
\newcommand{\efig}{\end{figure}}

\newcommand{\bemq}{\begin{quote} \begin{em}}
\newcommand{\eemq}{\end{em} \end{quote}}

\newcommand{\bmp}{\begin{minipage}}
\newcommand{\emp}{\end{minipage}}

\newcommand{\eqn}[1]{(\ref{#1})}

\newcommand{\brac}[1]{\left({#1}\right)}

\newcommand{\floor}[1]{\left\lfloor{#1}\right\rfloor}
\newcommand{\ceil}[1]{\left\lceil {#1} \right\rceil}

\newcommand{\kth}{^{{\mathrm{th}}}}

\newcommand{\define}{\triangleq}
\newcommand{\tendsto}{\to}

\newcommand{\ie}{{\it i.e.}}

\newcommand{\iid}{{i.i.d.}}

\newcommand{\expect}[1]{{\bf E}\left[{#1}\right]}

\newcommand{\prob}[1]{\text{Pr}\brac{#1}}

\newcommand{\bsp}{\begin{slide*}}
\newcommand{\esp}{\end{slide*}}
\newcommand{\bsl}{\begin{slide}}
\newcommand{\esl}{\end{slide}}

\IEEEoverridecommandlockouts

\newcommand{\EX}[1]{\expect{X_{#1}}}

\newcommand{\EXQ}[2]{\expect{X_{#1}^{(#2)}}}

\newcommand{\aveslotsQ}[2]{m_{#1}^{(#2)}}
\newcommand{\aveslots}[1]{m_{#1}}

\newcommand{\success}{{\cal S}}
\newcommand{\PI}[1]{P_{#1}}
\newcommand{\PLi}[1]{P_{L,#1}}
\newcommand{\PRi}[1]{P_{R,#1}}

\newcommand{\hsplit}[1]{\text{split}\brac{#1}}

\begin{document}

\title{Splitting Algorithms for Fast Relay Selection: Generalizations, Analysis, and a Unified View}

\author{Virag Shah, {\it Student Member, IEEE}, Neelesh B.  Mehta, {\it Senior Member, IEEE}, Raymond Yim, {\it Member, IEEE}
\thanks{V.\ Shah and N.\ B.\ Mehta are with the Electrical
    Communication Engineering Dept. at the Indian Institute of Science
    (IISc), Bangalore, India.  R.\ Yim is with the Mitsubishi Electric
    Research Labs (MERL), Cambridge, MA, USA.}
\thanks{Emails: \{\tt
    virag4u@gmail.com, nbmehta@ece.iisc.ernet.in, yim@merl.com\}.}
\thanks{A portion of this work has appeared in the IEEE International Conference on Communications (ICC) 2009.} 
 }

 \maketitle

\begin{abstract}
  
  Relay selection for cooperative communications promises significant
  performance improvements, and is, therefore, attracting considerable
  attention. While several criteria have been proposed for selecting
  one or more relays, distributed mechanisms that perform the
  selection have received relatively less attention. In this paper, we
  develop a novel, yet simple, asymptotic analysis of a
  splitting-based multiple access selection algorithm to find the
  single best relay. The analysis leads to simpler and alternate
  expressions for the average number of slots required to find the
  best user. By introducing a new `contention load' parameter, the
  analysis shows that the parameter settings used in the existing
  literature can be improved upon. New and simple bounds are also
  derived. Furthermore, we propose a new algorithm that addresses the
  general problem of selecting the best $Q \ge 1$ relays, and analyze
  and optimize it.  Even for a large number of relays, the algorithm
  selects the best two relays within 4.406 slots and the best three
  within 6.491 slots, on average. We also propose a new and simple
  scheme for the practically relevant case of discrete metrics.
  Altogether, our results develop a unifying perspective about the
  general problem of distributed selection in cooperative systems and
  several other multi-node systems.
\end{abstract}

\begin{keywords}
Relays, cooperative communications, selection, multiple access, splitting.
\end{keywords}

\IEEEpeerreviewmaketitle

\section{Introduction}

Selection mechanisms arise in many wireless communication schemes that
use most suitable candidates from among a set of many candidates. A
pertinent example is a cooperative communication system that exploits
spatial diversity by selecting the best relay(s) to forward a message
from a source to a destination. Selection makes cooperation practical
because it mitigates the tight synchronization that is required
among many geographically distributed cooperating
relays~\cite{bletsas_jsac_2006,lin_Globecom_2005,nam_CISS_2008,yang_VT_2008,luo_VTC_2005,yang_CISS_2006,madan_TWC_2008,beres_TWC_2008,michalopoulos_2008_TWC,lo_VT_2009_SubsetSelection,lo_VT_2009_HARQ}.
Another example is a cellular system that schedules in a 
proportional fair manner to the best mobile station based on 
the average data rate and the
current state of the channel between the base station and the
mobiles~\cite{tse_2005}. QoS requirements can also be incorporated in
the selection metric, as is done, for example, in a wireless local
area network (WLAN). In sensor networks, node selection is known to
improve network lifetime.

Several relay selection criteria have been proposed and analyzed in
the literature.  For example, \cite{bletsas_jsac_2006} showed that for
a decode-and-forward cooperation scheme, best relay selection achieves
full diversity.  In~\cite{nam_CISS_2008}, criteria for selecting
multiple relays were proposed to minimize data transmission time.
In~\cite{lo_VT_2009_SubsetSelection} relay subset selection was
considered for rate maximization.  In~\cite{yang_CISS_2006}, best two
relay selection was used to improve the diversity-multiplexing
tradeoff of an amplify and forward protocol.
In~\cite{madan_TWC_2008}, multiple relay selection was optimized for
cooperative beamforming. Multiple relay selection for wireless network
coding was considered in~\cite{ding_TWC_2009}.

The design of the mechanism that physically selects -- as per the
selection or suitability criteria -- the best relay or, in general,
the $Q$ best relays is, therefore, an important problem. Depending on
the transmission scheme, the suitability metric can be a function of
both the source-relay and relay-destination channel gains or just the
relay-destination or source-relay channel gains. It is desirable that
the mechanism be distributed since, typically, the knowledge of the
metric is initially available only locally at the relay.  For example,
a centralized polling mechanism for selection is undesirable as the
time to select increases linearly with the number of available relays.
To this end, a decentralized back-off timer-based scheme for single
best relay selection were proposed in~\cite{bletsas_jsac_2006}. In it,
each node transmits a short message when its timer expires. Making the
timer value inversely proportional to the metric ensures that the
first node that the sink hears from is the best node.  A distributed
single relay selection algorithm was also proposed
in~\cite{krikidis_CL_2007} to minimize the bit error rate.
In~\cite{liu_2006_TWC}, the source uses handshake messages from relays
to track the rate that each candidate relay can support.

An alternate approach considers a time-slotted multiple access
contention based algorithm in which each active node locally decides
whether or not to transmit in a certain time slot.  Recently,
variations based on splitting algorithms, which were extensively
researched two decades ago for multiple access control~\cite[Chp.\ 
4]{bertsekas_gallager}, have been proposed for single relay
selection~\cite{qin_infocomm_2004,yim_TWC_2009_VPMAS}. In each step of
the splitting-based selection algorithm proposed
in~\cite{qin_infocomm_2004}, only those nodes whose metrics lie
between two thresholds transmit. The nodes update the thresholds
(independently) in each slot based on the outcome of the previous slot
fed back by the sink.\footnote{We use the generic term `sink' to refer
  to the source or access point or base station, as the case may be,
  that needs to select the best node/relay.} It was shown
in~\cite{qin_infocomm_2004} for continuous metrics that the best node
can be found, on average, within at most 2.507 slots even for an
infinite number of nodes.  This result was obtained by deriving an
upper bound on the average number of slots when the number of relays
tends to infinity. However, the analysis was quite involved and the
upper bound was in the form of an infinite series.

While distributed selection mechanisms have proposed for single relay
selection, several questions remain open. For example, developing a
comprehensive analysis of the splitting mechanism is an important
problem. A natural question that such an analysis will answer is how
to optimally choose the thresholds to improve the speed of selection.
In~\cite{qin_infocomm_2004}, the thresholds are initially set greedily
so to maximize the probability of success. As we show, this is not
optimal. Furthermore, efficient mechanisms are yet to be developed for
multiple relay selection. The only option known currently is to run
the single relay selection algorithm multiple times, which, as we show
in this paper, is inefficient. Finally, the mechanisms above assume
that the selection metric is continuous, and exploit the fact that,
with probability 1, no two relays have the same metric. The mechanism
catastrophically breaks down when the metrics are discrete, which can
often occur in practice.  This occurs, for example, when the
estimation inaccuracy renders higher resolution representations
unnecessary, or when quantized metrics for feedback or QoS are
considered~\cite{lo_VT_2009_HARQ,choumas_LANMAN}.

This paper thoroughly examines splitting-based selection algorithms
for both continuous and discrete metrics, and makes the following
significant contributions:
\begin{itemize}
  
\item {\em Analysis of single relay selection:} The paper develops a
  novel and considerably simpler {\em exact} asymptotic analysis for a
  general version of the splitting algorithm.  It achieves this by
  developing a different Poisson process interpretation of the metric
  distribution, which has not been used before to the best of our
  knowledge. Furthermore, it also derives a new convex and simple
  upper bound for the average number of slots required to select the
  best relay.

\item {\em Optimization of single relay selection:} The paper
  analytically determines the optimal performance of the splitting
  algorithm. It also rigorously shows that the greedy parameter choice
  of~\cite{qin_infocomm_2004} is sub-optimal, but is still very good.
  
\item {\em An alternate Markovian analysis:} The Poisson process
  interpretation also leads to an alternate and novel Markovian
  analysis, which among other things yields a new exact asymptotic
  expression for the average number of slots. As we shall see, while
  the two new expressions derived in this paper are equivalent, they
  exhibit different behaviors when truncated.
  
\item {\em New mechanism for multiple relay selection, including its
    analysis and optimization:} The paper proposes a novel scalable,
  fast, and decentralized algorithm for the general problem of
  selecting not just the single best but the best $Q\geq1$ relays. To
  the best of our knowledge, this is the fastest family of $Q$ relay
  selection algorithms proposed to date. We develop an asymptotic
  analysis of the general $Q$ relay selection algorithm, and determine
  its optimal parameters. We show that as $Q$ increases, the greedy
  parameter choice becomes more suboptimal. In effect, as $Q$
  increases, the optimal splitting algorithm prefers that more nodes
  collide since it is faster to resolve a collision than avoid one.

\item {\em Unifying perspective:} The paper shows that the optimized
  best relay selection algorithm, the proposed multiple relay
  selection mechanism, and Gallager's First Come First Serve (FCFS)
  multiple access control algorithm~\cite{bertsekas_gallager} are intimately
  related.
  
\item {\em New scalable algorithm for discrete metrics:} Finally, the
  paper proposes a novel, scalable, and an intuitive distributed
  scheme called {\em Proportional Expansion}, which enables the single
  and multiple relay selection algorithms to be applied to the
  practical case of discrete metrics.

\end{itemize}

The rest of the paper is organized as follows.  The analysis and
results for single best node selection is developed in
Sec.~\ref{sec:Single Relay selection}.  The new algorithm for $Q\geq
1$ node selection is proposed, analyzed, and simulated in
Sec.~\ref{sec:Q-Relay Selection Algorithm}.  We conclude in
Sec.~\ref{sec:conclusions}.  Several mathematical proofs are relegated
to the Appendix.

\section{Single Relay selection}
\label{sec:Single Relay selection}

\subsection{System Setup}
\label{subsec:single relay algo}

Consider a time-slotted system with $n$ active nodes and a sink, as
shown in Fig.~\ref{fig:system}. Each node $i$ has a suitability metric
$u_i$, which is known only to that specific node.  In this section,
the goal is to select the node with the highest metric.  The metrics
are continuous and i.i.d.\ with complementary CDF (CCDF) denoted by
$F_c(u) = \Pr(u_{i} > u)$.  Therefore, the $F_c(.)$ is monotonically
decreasing and invertible. (The discrete metric case, where this is
not so, is tackled in Sec.~\ref{sec:discrete_metric}.)

\begin{figure}[p]
\centering
\input{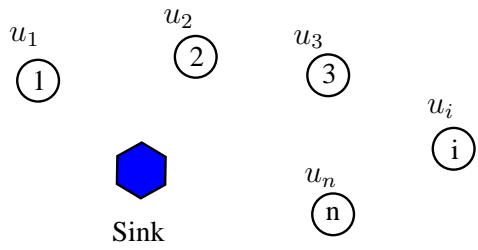}
\caption{A relay selection system consisting of a sink and $n$ relays/nodes, with a node $i$ possessing a suitability metric $u_{i}$.}
\label{fig:system}
\end{figure}

\subsection{Splitting Algorithm: Brief Review and Notation}

We now formally define the splitting algorithm for single relay
selection. To keep the treatment concise, we first define the state
variables maintained by the algorithm and their initialization.
Thereafter, we describe how the algorithm controls the transmissions
of the nodes, how the sink generates feedback based on these
transmissions, and how the state variables get autonomously updated
based on the feedback.

{\em Definitions:} The generalized best relay selection algorithm is
specified using three variables $H_L(k)$, $H_H(k)$ and $H_{\min}(k)$;
the notation being consistent with that in~\cite{qin_infocomm_2004}.
$H_L(k)$ and $H_H(k)$ are the lower and upper metric thresholds such
that a node $i$ transmits at time slot $k$ only if its metric $u_i$
satisfies $H_L(k) < u_i < H_H(k)$.  $H_{\min}(k)$ tracks the largest
value of the metric known up to slot $k$ above which the best metric
surely lies.

{\em Initialization:} In the first slot ($k = 1$), the parameters are
initialized as follows: $H_L(1) = F_{c}^{-1}(p_e/n)$, $H_H(1) =
\infty$, and $H_{min}(1) = 0$. Here, $p_e$ is a system parameter, and
shall henceforth be referred to as the {\em Contention load parameter}.

{\em Transmission rule:} At the beginning of each slot, each node
locally decides to transmit. As mentioned, it transmits if and only if
its metric lies between $H_L(k)$ and $H_H(k)$.

{\em Feedback generation:} At the end of each slot, the sink
broadcasts to all nodes a two-bit feedback: (i)~$0$ if the slot was idle (when no node
transmitted), (ii)~$1$ if the outcome was a success (when exactly one
node transmitted), or (iii)~$e$ if the outcome was a collision (when
multiple nodes transmitted).\footnote{The sink can distinguish between
  these outcomes using, for example, the strength of the total
  received power~\cite{yim_2008_ICC}.}

{\em Response to feedback:} We first define the split
function\footnote{The split function makes sure that on an average
  half of the nodes involved in the last collision transmit in the
  next slot.  Splitting can be made faster as was done
  in~\cite{cohen_1995_Cybernetics}. However, doing so requires each
  node to numerically calculate thresholds in each slot that are
  solutions of degree $n-1$ equations. Also, the improvement due to
  this scheme turns out to be less than $0.5\%$.}  to facilitate
description: Let $\hsplit{a,b} =
F_c^{-1}\brac{\frac{F_c(a)+F_c(b)}{2}}.$ Then, depending on the
feedback, the following possibilities occur:
\begin{enumerate}
\item If the feedback (of the $k\kth$ slot) is an idle ($0$) and no collisions has occurred so far, then set $H_H(k+1) = H_L(k)$,
  $H_L(k+1) = F_{c}^{-1}(\frac{k+1}{n}p_e)$, and  \mbox{$H_{\min}(k+1) =  0$}.

\item If the feedback is a collision ($e$), then set \mbox{$H_L(k+1) =
    \hsplit{H_L(k),H_H(k)}$}, $H_H(k+1) = H_H(k)$, and
  $H_{\min}(k+1) = H_L(k)$

\item If the feedback is an idle ($0$) and a collision has occurred in
  the past, then set $H_H(k+1) = H_L(k)$, $H_L(k+1) =
  \hsplit{H_{\min}(k), H_L(k)}$, and $H_{\min}(k+1) = H_{\min}(k)$.
\end{enumerate}
{\em Termination:} The algorithm terminates when the outcome is a
success ($1$).

We shall call the durations before and after the first non-idle slot
as the {\it idle} and {\it collision} phases, respectively.  Thus, the
contention load parameter, $p_e$, is the average number of users that
transmit in a slot in the idle phase. The Qin-Berry
algorithm~\cite{qin_infocomm_2004} uses $p_e = 1$, which is the value
that maximizes the probability of a success outcome in an idle phase
slot.

\subsection{Main Analytical Results}

The floor and ceil operations are denoted by $\floor{.}$ and
$\ceil{.}$, respectively. $\expect{Z}$ will denote the expected value
of a random variable $Z$.

We now develop a new analysis of the average time taken,
$\aveslots{n}(p_e)$, by the splitting algorithm to select the single
best relay. The following lemma gives an exact expression for
$\aveslots{n}(p_e)$.
\begin{lemma}
\label{average_no_of_slots_finite_n}
Let $X_k$ be the number of slots required to resolve a collision among
$k$ nodes.  Let $q = \ceil{\frac{n}{p_e}}-1$ denote the idle phase
duration in slots. The average number of slots, $\aveslots{n}(p_e)$,
required to find the best node is given by,
\begin{equation} \label{eq:finite_nodes}
 \aveslots{n}(p_e) =\sum_{i=1}^q\sum_{k=1}^n  \binom{n}{k}  \left(\frac{p_e}{n}\right)^k
  \left(1-\frac{ip_e}{n}\right)^{n-k}  \left(\EX{k}+i\right)
  + \left(1-\frac{qp_e}{n}\right)^n (\EX{n} + q+1),
\end{equation}
where $\EX{k}$ follows the recursion $\EX{k} =
\frac{0.5^k\left(\sum_{l=2}^{k-1}
    \binom{k}{l}\EX{l}\right)+1}{1-0.5^{k-1}}$, for all $k \ge 2$, and
$\EX{1} = 0$.
\end{lemma}
\begin{proof}
The proof is given in Appendix~\ref{proof of average_no_of_slots_finite_n}.
\end{proof}

The above expression is complex and does not directly reveal the
scalable nature of the algorithm.  The theorem below provides two
equivalent and new expressions for the asymptotic case ($n \tendsto
\infty$).  
 \begin{theorem}
  \label{m1}
  The average number of slots required to find the best node as
  $n\tendsto \infty$ is given by following two different yet
  equivalent expressions.
\begin{enumerate}
\item {\em Recursive expression:}
 \begin{equation}
 \aveslots{\infty}(p_e)  = \frac{1}{e^{p_e}-1}\sum_{k=1}^\infty \frac{\EX{k} p_e^k}{k!}  + \frac{1}{1-e^{-p_e}}.
\label{eq:m1}
 \end{equation}

\item {\em Non-recursive expression:}
\begin{equation}
\aveslots{\infty}(p_e)=\frac{1}{1-e^{-p_e}} + \sum_{i=1}^\infty p(i),
\label{eq:m2}
\end{equation}
where $p(i)= (1-P_0)\prod_{j=1}^{i-1} (1-\PI{j})$,
$P_0=\frac{p_e e^{-p_e}}{1-e^{-p_e}}$,
\mbox{$\PI{i}=\frac{2^{-i}p_ee^{-2^{-i}p_e}(1-e^{-2^{-i}p_e})}{1-(1+2^{-(i-1)}p_e)e^{-2^{-(i-1)}p_e}}$},~$\forall~i
\ge 1$.
\end{enumerate}
 \end{theorem}
\begin{proof}
  We show the proof for the recursive expression in~\eqn{eq:m1} below
  as it leads to a powerful new Poisson point process interpretation
  that will be useful throughout this paper.  For example, it will
  lead to the derivation of the non-recursive expression
  in~\eqn{eq:m2}, whose proof is relegated to Appendix~\ref{proof of
    m2}. The physical meaning of $p(i)$ and $P_{j}$ will become clear
  after the proof.

  Let node $i$ have metric $u_i$ with CCDF $F_c(u)$. Let $y_i =
  nF_c(u_i)$. Then, $y_i$ are i.i.d.\ and are uniformly distributed in
  $[0,n]$.  {\it Note that selecting the node with the highest $u_{i}$
    is equivalent to selecting the node with the lowest $y_{i}$
    because the CCDF is a monotonically decreasing function.}  Sorting
  $\{y_{i}\}_{i=1}^n$ in ascending order, we get \mbox{$y_{[1]} \le
    y_{[2]} \le y_{[3]} \cdots \le y_{[n]}$}, where $[i]$ is the index of
  the relay with the $i^{\text{th}}$ largest metric.

  Given $y_{i}$, we can define a point process~\cite{wolff} $M(t)$ as
  \mbox{$M(t)=\max{\left\{k\ge1:y_{[k]}\leq t\right\}}$}. Thus, $M(t)$
  is the number of points that have occurred up to time $t$. Since
  $\{y_{i}\}_{i=1}^n$ are \iid\ and uniformly distributed, $M(t)$ is
  binomially distributed.  As $n \to \infty $, it can be shown that
  $M(t)$ forms a Poisson process with rate 1~\cite{wolff}. Now, the
  probability that the first non-idle slot is the $i\kth$ slot and $k
  \ge 1$ nodes are involved is equal to the probability that
  $y_{[1]},\ldots,y_{[k]}$ lie between $(i-1)p_{e}$ and $i p_{e}$, and
  $y_{[j]} > ip_{e}$, for $k+1 \leq j \leq n$.  It also implies that
  no points lie between $0$ and $(i-1)p_{e}$.  Therefore,
\begin{align}
&\hspace{-1cm}\prob{ x_{[1]}>(i-1)\frac{p_e}{n},(i-1)\frac{p_e}{n}<x_{[k]}<i\frac{p_e}{n},x_{[k+1]}>i\frac{p_e}{n}}\nonumber\\
& = \hspace{3mm}\prob{ M((i-1)p_e)=0, \; M(ip_e)=k},\nonumber\\
& = \hspace{3mm}\prob{ M((i-1)p_e)=0 }\prob{ M(ip_e)=k \;|M((i-1)p_e)=0},\nonumber\\
& \stackbin[n \tendsto \infty]{a}{=}  \hspace{3mm} e^{-(i-1)p_e}\; e^{-p_e}\frac{p_e^k}{k!} = e^{-ip_e}\frac{p_e^k}{k!}.
\end{align}
Here, (a) follows from the memoryless property of the Poisson
process~\cite{wolff}.  Recall that $\EX{k}$ is the expected number of
slots required to resolve a collision among $k$ nodes.  Thus, if the
first non-idle slot is the $i\kth$ slot and $k \ge 1$ nodes are
involved, then $\EX{k}+i$ slots are required to find the best node.
Also, as $n\to\infty$, $q p_e/n\to 1$. Hence, we get
$\aveslots{\infty}(p_e) = \sum_{i=1}^\infty\sum_{k=1}^\infty
e^{-ip_e}\frac{p_e^k}{k!} \left(\EX{k}+i\right)$.
%
The desired
result follows with the help of combinatorial
identities~\cite{gradshteyn00_book}.
\end{proof}

%

The main theorem readily gives rises to the following upper bound expression that does not involve an infinite series.  
\begin{corollary} \label{upperbound}
For any real $k_0 \geq e/2$,
\begin{equation} \label{eq:upperbound}
\aveslots{\infty}(p_e) \le \frac{p_e}{k_0\log_{e}(2)} + \log_2\brac{\frac{2 k_0}{e}} + \frac{1}{1-e^{-p_e}}.
\end{equation}

\end{corollary}
\begin{proof}
The proof is given in Appendix~\ref{proof of upperbound}.
\end{proof}

Alternatively, since both the expressions derived in Theorem~\ref{m1}
involve only positive terms in the series summation, considering only
the first few terms of the infinite series in~\eqn{eq:m1} and
\eqn{eq:m2} results in tight lower bounds near the optimal contention
parameter value.  These simplified expressions allow system designers
to quickly compute the necessary parameters for system optimization.
As we shall see, their behavior turns out to be quite different and
sheds light on the differences between the two equivalent expressions
derived in Theorem~\ref{m1}.

\subsection{Results for Single Relay Selection}

Figure~\ref{fig:m(p_e)} plots the average number of slots required to
select the best node as a function of $p_{e}$ for the two expressions
and verifies them using Monte Carlo simulations.  It can be seen that
the asymptotic expression is accurate even when the number of relays
is small (e.g., 10). Furthermore, the optimal value of
$\aveslots{\infty}(p_{e})$ is 2.467, and occurs at $p_e=1.088$. As
expected, the optimal $p_e$ does not exceed 2. This is because having
more than two nodes on average to transmit and collide in a slot is
suboptimal. We also observe that $\aveslots{\infty}(p_{e})$ at $p_e =
1$ is quite close to the optimal value.
\begin{figure}[p]
        \centering
                \includegraphics[width=\linewidth]{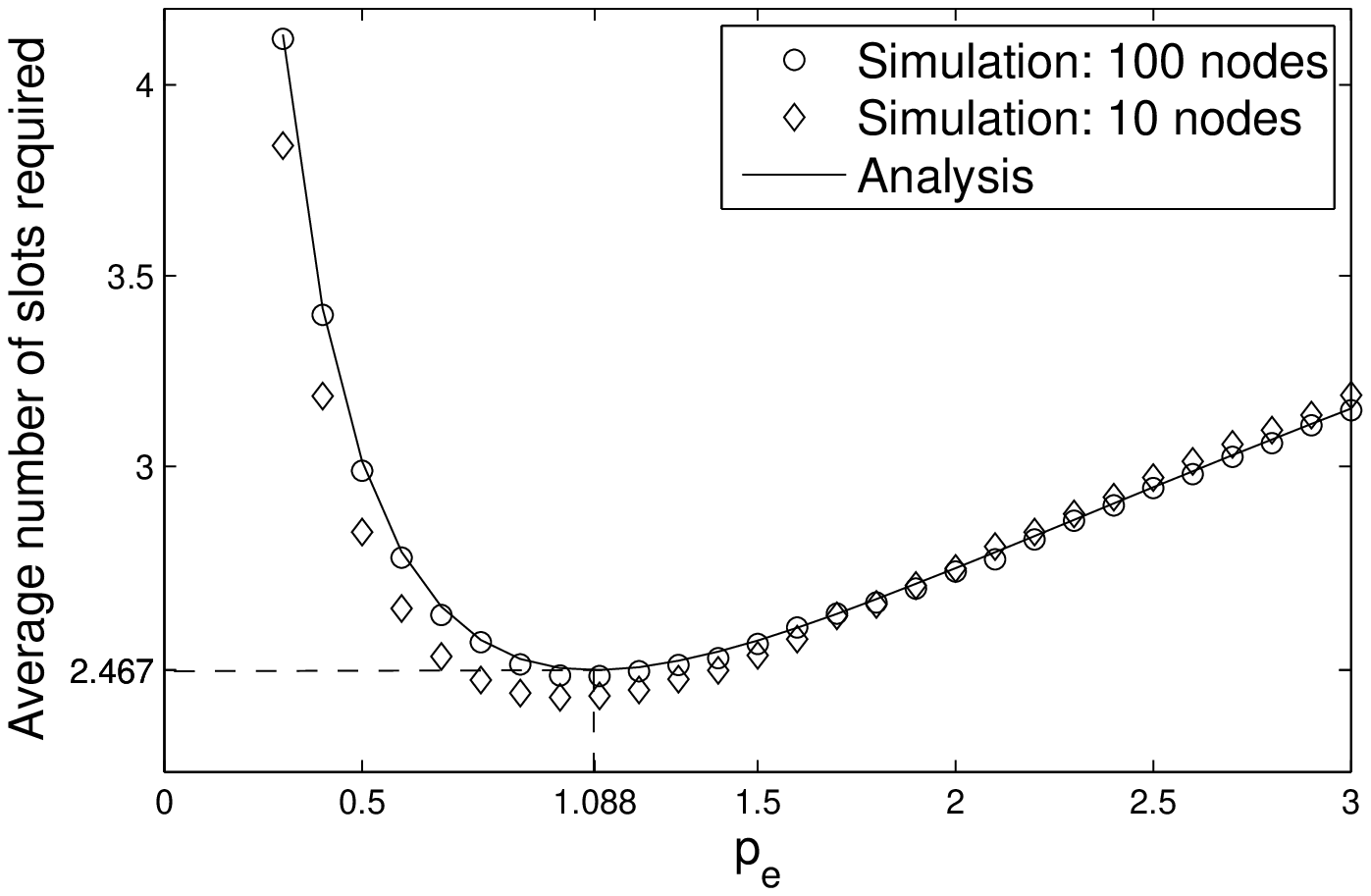}
                \caption{Average number of slots required to select the best node $(\aveslots{\infty}(p_e))$ as a function of $p_e$ }
        \label{fig:m(p_e)}
\end{figure}

Figure~\ref{fig:upperbounds} plots the upper bound using $k_0 = 2$. As
expected, it has a unique minimum and follows the behavior of the
exact expression well in the region of interest of $p_e$.  The same
figure also compares the lower bounds obtained using the first 4 terms
of both the expressions of Theorem~\ref{m1}.  For higher values of
$p_e$, the lower bound obtained by truncating the recursive expression
in~\eqn{eq:m1} does not capture the behavior of the exact expression
well. This is because of the truncation, on account of which
the possibility that a large number of nodes collide in the first non-idle
slot is not accounted for. This probability is not negligible for
larger $p_e$.  However, the lower bound obtained by truncating the
non-recursive expression in~\eqn{eq:m2} does better at larger $p_e$
because the summation in the series is over the number of slots
required after the first non-idle slot and not over the number of
nodes that collided in the first non-idle slot.

%
%
%
\begin{figure}
  \centering \includegraphics[width=0.8\columnwidth,
  keepaspectratio]{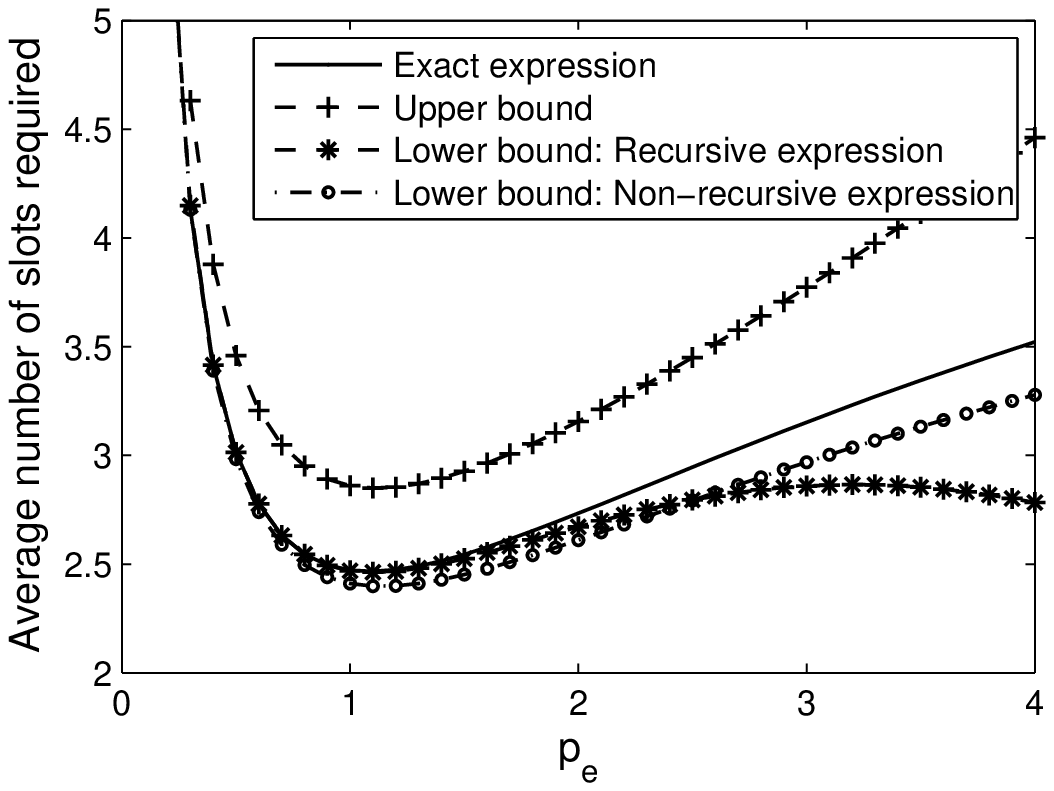}
                \caption{Upper and lower bounds for the average number of slots required to select the best node.}
        \label{fig:upperbounds}
\end{figure}

\section{Q-Relay Selection Algorithm}
\label{sec:Q-Relay Selection Algorithm}

We now develop a new family of splitting algorithms for selecting the
relays with the $Q$ best (highest) metrics, where $Q$ is a
pre-specified system parameter. The value of $Q$ depends on the system
under consideration. For example, in~\cite{ding_TWC_2009}, $M-1$
relays need to be selected to forward the transmissions by $M$
sources.  The choice of $Q$, which is beyond the scope of this paper,
is ultimately governed by the end-to-end system performance and
practical constraints such as the synchronization requirements across
the selected relays. For example, while having more cooperative relays
improves the reliability or speed of transmission of data to the
destination, selecting them will also require the system to expend
additional resources. The reader is referred
to~\cite{nam_CISS_2008,madan_TWC_2008,lo_VT_2009_SubsetSelection,shah_globecom_2009}
for a detailed discussion on this aspect.

\subsection{Algorithm Motivation and Definition}
\label{subsec:Q-algo defiition}

When we revisit the asymptotic regime considered in the previous
section, we observe the following.  The single node selection
algorithm, in effect, runs the FCFS
algorithm~\cite{bertsekas_gallager} on the Poisson point process
$M(t)$ defined in Theorem~\ref{m1}, with $t$ being interpreted as
time.  However, unlike FCFS, the single relay selection algorithm
stops as soon as it finds the first (best) node.  In this context, the
parameter $p_e$ is analogous to FCFS's {\em initial contention
  interval}.

Based on the above insight, we now formally state the new multi-relay
selection algorithm given any $Q$. We then explain the logic behind it
and fully analyze it.  For this, we adopt the notation used for FCFS
in~\cite{bertsekas_gallager}, as it turns out to be more convenient.

As in Sec.~\ref{sec:Single Relay selection}, let $y_{i} = n F_c(u_i)$.
The algorithm specifies four state variables $S(k)$, $T(k)$,
$\alpha(k)$, and $\sigma(k)$ for each slot $k$. $S(k)$ is the number
of nodes selected before slot $k$.  $(T(k), T(k)+\alpha(k))$
represents the {\em threshold interval} for slot $k$, $\ie$, all the
nodes with $y_i\in(T(k), T(k)+\alpha(k))$ transmit in slot $k$.
(Equivalently, $H_H(k)=F_c^{-1}\brac{T(k)/n}$ and \mbox{$H_L(k)=
  F_c^{-1}\brac{ \brac{ T(k) + \alpha(k)}/n}$.}
$\sigma(k)\in\{L,R\}$ indicates whether the $k\kth$ slot interval is
the left half or the right half of the previously split interval.
During initial slots, when no collision is to be resolved,
$\sigma(k)=R$ by convention. Thus, for $k = 1$, we have
$S(1) = 0$, $T(1) = 0$, $\alpha(1)=p_e$, and $\sigma(1) = R$.

In the $(k+1)\kth$ slot ($k \geq 1$):
\begin{enumerate}
\item If feedback is a collision ($e$), then $T(k+1) = T(k) $, \mbox{$\alpha(k+1) = \alpha(k)/2$,} and $\sigma(k+1) = L$.

\item If feedback is a success ($1$) and $\sigma(k) = L$, then \mbox{$T(k+1) = T(k) + \alpha(k)$}, $\alpha(k+1) = \alpha(k)$,
and \mbox{$\sigma(k+1) = R$. }

\item If feedback is an idle ($0$) and $\sigma(k) = L$, then \mbox{$T(k+1) = T(k) +\alpha(k)$,} $\alpha(k+1) = \alpha(k)/2$, and $\sigma(k+1) = L$.

\item If feedback is an idle ($0$) or a success ($1$), and $\sigma(k) = R$, then $ T(k+1) = T(k)+\alpha(k)$, $\alpha(k+1) = p_{e}$, and $\sigma(k+1) = R$.

\item Increment $S(k+1)$ by 1 if feedback is a success ($1$). Terminate if $S(k+1)$ reaches $Q$.
\end{enumerate}

\subsection{Brief Explanation}

The logic behind the algorithm is as follows: (i)~When a collision
occurs, the threshold interval for the next slot is the left ($L$)
half of that of the present slot. (ii)~When a collision occurs, the
threshold interval must have at least 2 nodes. Thus, when a success
follows a collision, the threshold interval for the next slot is the
right (higher) half ($R$) of the previous slot, since it is known to
have at least one node.  (iii)~When an idle follows a collision, it
implies that all the nodes involved in collision lie in the right half
of the previous split interval. Thus, it is further split it into two
equal halves, and the threshold interval for the next slot is the left
half of this split.  (iv)~When there is no collision to be resolved,
the algorithm moves to the adjacent threshold interval (which we call
as collision resolution interval) of size $p_e$.\footnote{We can relax
  the restriction that each collision resolution interval is of length
  $p_e$. However, it can be shown that doing this leads to a
  negligible improvement.} As mentioned above, the algorithm
terminates after the $Q$ successes.

{\em Comments:} The proposed algorithm is equivalent to the algorithm
of Sec.~\ref{sec:Single Relay selection} when $Q=1$. It is similar to
FCFS, except that it stops after the $Q\kth$ success. There is one
subtle difference, however, between the algorithm and FCFS. In FCFS,
the contention resolution interval can be smaller than $p_e$ if the
difference between the current time and the time of the last resolved
interval is small.  However, this does not happen in our algorithm
(step~4) because all the nodes know their individual metrics a priori.
Notice that the algorithm is greedy in that it does not account for
possible interactions between metrics of the relays.  However, such a
greedy approach has often been used given its inherent
distributability~\cite{ding_TWC_2009}.\footnote{A more general version
  of the algorithm would allow for the metrics to be modified on the
  basis of the relays that have already been selected.  Developing
  such an algorithm is an interesting avenue for future work, and
  would find several applications, such as in the time-sharing
  proportional fair solution of~\cite{oechtering_TWC_2008}.}

\subsection{Algorithm Analysis:  Best Two Nodes Selection}

First, we analyze the algorithm for selecting the best two nodes using
the Poisson point approach that came out of Sec.~\ref{sec:Single Relay
  selection}.  This will lead to an analysis for the general $Q > 2$
node selection case. The $Q = 2$ analysis is shown separately as it
turns out to be richer.

Let $\aveslotsQ{\infty}{Q}(p_e)$ represent the average number of slots
required to select the best $Q$ nodes. Thus, the symbol
$\aveslots{\infty}(p_e)$, which was used in the previous section on
single relay selection, is equivalent to $\aveslotsQ{\infty}{1}(p_e)$.
The following theorem gives two different but equivalent and exact
expressions for $\aveslotsQ{\infty}{2}(p_e)$.
\begin{theorem}
\label{2 best nodes}
Let $\EXQ{k}{Q}$ denote the average number of slots required to select
the best $Q$ nodes after $k$ nodes collide. As $n\to \infty$,
$\aveslotsQ{\infty}{2}(p_e)$ is given by
\begin{equation}\label{aveslots2}
\aveslotsQ{\infty}{2}(p_e) =\frac{1}{e^{p_e}-1} \sum_{k=1}^\infty \frac{\EXQ{k}{2}p_e^k}{k!}  + \frac{1}{1-e^{-p_e}},
 \end{equation}
where
\begin{equation}
\label{EX_k[2]}
 \EXQ{k}{2}=
\left(2^{k}-2\right)^{-1}\Bigg(\bigg(\sum_{i=2}^{k-1} \binom{k}{i}\EXQ{i}{2} + k\left(1+\EXQ{k-1}{1}\right) \bigg)+2^k\Bigg), \: \forall~k \ge 3,
 \end{equation}
 $ \EXQ{2}{2}=3$, and $\EXQ{1}{2}=\aveslotsQ{\infty}{1}(p_e)$.

Alternately, $\aveslotsQ{\infty}{2}(p_e)$ also equals \begin{equation}
\aveslotsQ{\infty}{2}(p_e)=\frac{1}{1-e^{-p_e}}  + P_0\aveslotsQ{\infty}{1}(p_e) + \sum_{i=1}^\infty \left( p(i)+p'(i)+p''(i+1)\right),
\label{aveslots2_2}
 \end{equation}
 where
$p(i)= (1-P_0)\prod_{j=1}^{i-1} (1-\PLi{j})$, $\forall~i\ge1$,
$p'(i)=p(i)\PLi{i},\forall~i\ge1$,
 $p''(2)=p'(1)(1-\PRi{1})$, and $p''(i)=p'(i-1)(1-\PRi{i-1})+p''(i-1)(1-\PLi{i-1})$, $\forall~i>2$.
 Here, $P_0=\frac{p_e e^{-p_e }}{1-e^{-p_e }},$
$ \PLi{i}=\frac{2^{-i}p_ee^{-2^{-i}p_e}(1-e^{-2^{-i}p_e})}{1-(1+2^{-(i-1)}p_e)e^{-2^{-(i-1)}p_e}}, \;\forall i\ge 1$,
$ \PRi{i}=\frac{2^{-i}p_ee^{-2^{-i}p_e}}{1-e^{-2^{-i}p_e}}, \;\forall i\ge 1$.
\end{theorem}
\begin{proof}
  The proof is given in Appendix~\ref{proof of 2 best nodes}. It also
  gives a physical meaning for $p(i)$, $p'(i)$, $p''(i)$, $\PLi{i}$,
  and $\PRi{i}$.
\end{proof}

\subsection{Algorithm Analysis:  Best $Q>2$ Nodes Selection}
We now derive a general expression for $\aveslotsQ{\infty}{Q}(p_e)$
for any $Q>2$ This generalizes the first result of Theorem~\ref{2 best
  nodes}.

\begin{theorem}
\label{Q best nodes}
As $n\to \infty$, the average number of slots
required to select the  best $Q>2$ nodes is
 \begin{equation}
\label{aveslotsQ}
\aveslotsQ{\infty}{Q}(p_e) =\frac{1}{e^{p_e}-1} \sum_{k=1}^\infty \frac{\EXQ{k}{Q}p_e^k}{k!}  + \frac{1}{1-e^{-p_e}},
 \end{equation}
where
 \begin{equation}
\label{EX_k[Q]}
 \EXQ{k}{Q}=
\left(2^{k}-2\right)^{-1}\Bigg(\bigg(\sum_{i=2}^{k-1} \binom{k}{i}\EXQ{i}{Q} + k\left(1+\EXQ{k-1}{Q-1}\right) \bigg)+2^k\Bigg), \: \forall~k \ge 3,
 \end{equation}
$\EXQ{2}{Q}=\aveslotsQ{\infty}{Q-2}(p_e) + 3, \; \forall \; Q > 2$, $ \EXQ{2}{2}=3$, and $\EXQ{1}{Q}\!\!=\!\aveslotsQ{\infty}{Q-1}(p_e)$.
\end{theorem}
\begin{proof}
The proof is given in Appendix~\ref{proof of Q best nodes}.
\end{proof}
A non-recursive expression for $\aveslotsQ{\infty}{Q}(p_e)$ for $Q>2$
along the lines of~\eqn{eq:m2} of Theorem~\ref{m1} and \eqn{aveslots2}
of Theorem~\ref{2 best nodes} can be derived.  However, the Markov
chains become more involved.
%
%
%

\subsection{Results for $Q$ Best Relay Selection}

\begin{figure}[p]
  \centering \includegraphics[width=0.8\columnwidth,
  keepaspectratio]{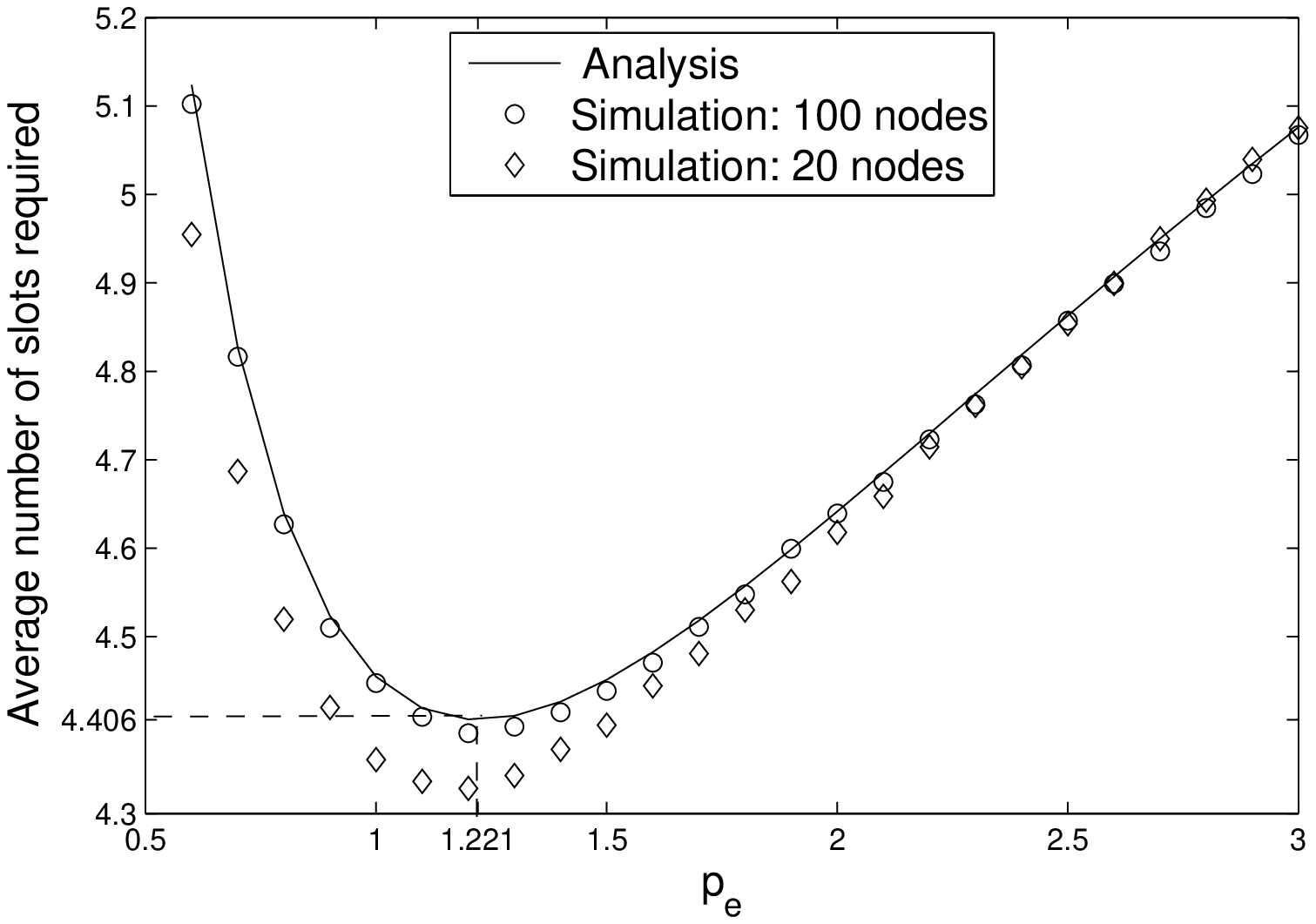}
\caption{Average number of slots required to select the  best two nodes~$(\aveslotsQ{\infty}{2}(p_e))$ as a function $p_e$.}
\label{fig:2_best_nodes}
\end{figure}

Figure~\ref{fig:2_best_nodes} plots $\aveslotsQ{\infty}{2}(p_e)$ as a
function of $p_e$ using Theorem~\ref{2 best nodes} and verifies it
using Monte Carlo simulations. It can be seen that the asymptotic
expressions are accurate even for a small number of nodes, e.g., $n =
20$. The lowest average number of slots required to select two users
is $4.406$, which occurs at $p_e=1.221$.  This is 10.7\% faster than
running the single relay selection algorithm twice, which requires $2
\times 2.467 = 4.934$ slots. The increase in the optimal $p_e$ from
$1.088$ slots for $Q = 1$ to $1.221$ slots for $Q = 2$ occurs because
now it is faster to resolve a collision than to avoid it.
Specifically, the time taken to select two nodes given that they are
involved in a collision is $\EXQ{2}{2}=3.0$ slots.  Where as, the
number of slots required to select two nodes, given that the previous
slot was idle, is $4.4$ slots.

\begin{table}[p]
\renewcommand{\arraystretch}{1.2}
\caption{Optimum $p_e$ and the average number of slots required to select the best $Q$ relays}
\centering
\begin{tabular}{| c | c | c | c |}
\hline
        $Q$ & Optimum $p_e$ & Optimum $\aveslotsQ{\infty}{Q}(p_e)$ (slots) & Improvement\\
        \hline
        1 & 1.088 & 2.467 & - \\
        2 & 1.221 & 4.406 & 10.7\% \\
        3 & 1.214 & 6.491 & 12.3\% \\
        4 & 1.231 & 8.537 & 13.5\% \\
        5 & 1.236 & 10.592 & 14.1\% \\
        6 & 1.241 & 12.645 & 14.6\% \\
        \hline
\end{tabular}
\label{table:optimum p_e}
\end{table}

Table~\ref{table:optimum p_e} provides the optimum values of $p_e$ and
the average number of slots as a function of the number of relays that
need to be selected. We can see that selecting the best three nodes
takes $6.491$ slots, on average, and is achieved when $p_e =
1.214$.\footnote{The marginal decrease in the optimal value of $p_e$
  from 1.221 to 1.214 when $Q$ increases from $2$ to $3$ can be
  explained as follows.  The time taken to select three nodes after a
  collision among two nodes is $\EXQ{2}{3}=5.48$ slots. However, the
  number of slots required to select three nodes after an idle slot,
  is $6.49$ slots, which is just $17.8\%$ more than $5.48$. Therefore,
  the optimum $p_e$ decreases since the selection times after an idle
  and a collision are not as unequal as for $Q = 2$.} As $Q \tendsto
\infty$, the optimum value of $p_e$ increases to $1.266$, which is
also the optimum value maximizing the throughput of
FCFS~\cite{bertsekas_gallager}.\footnote{The maximum arrival rate of
  0.487 is supported when initial collision interval is capped at 2.6.
  This implies that there are on average $0.487 \times 2.6 = 1.266$
  nodes transmitting.  The contention parameter $p_e$ is set using
  normalized metric CCDF with an `arrival rate' equal to 1.}  Also, it
can be shown that $\frac{Q}{\aveslotsQ{\infty}{Q}(1.266)}$, which
represents the average number of users selected per slot by the
algorithm for $p_e=1.266$, increases to $0.487$ as $Q \tendsto
\infty$.


\subsection{Tackling Discrete Metrics Using Proportional Expansion}
\label{sec:discrete_metric}

\newcommand{\unif}[1]{\mathcal{U}(#1)}

The thresholding algorithms in Sec.~\ref{subsec:single relay algo} and
Sec.~\ref{subsec:Q-algo defiition} exploit the critical fact that with
probability one no two metrics are equal.  However, as mentioned in
the Introduction, when the metric has a discrete probability
distribution, the algorithms break down because the probability that
the metrics of the best two nodes are exactly equal is non-zero.  We
now provide a simple and novel distributed solution called {\em
  Proportional Expansion} to tackle this practical problem.

\begin{figure}[p]
\centering
\input{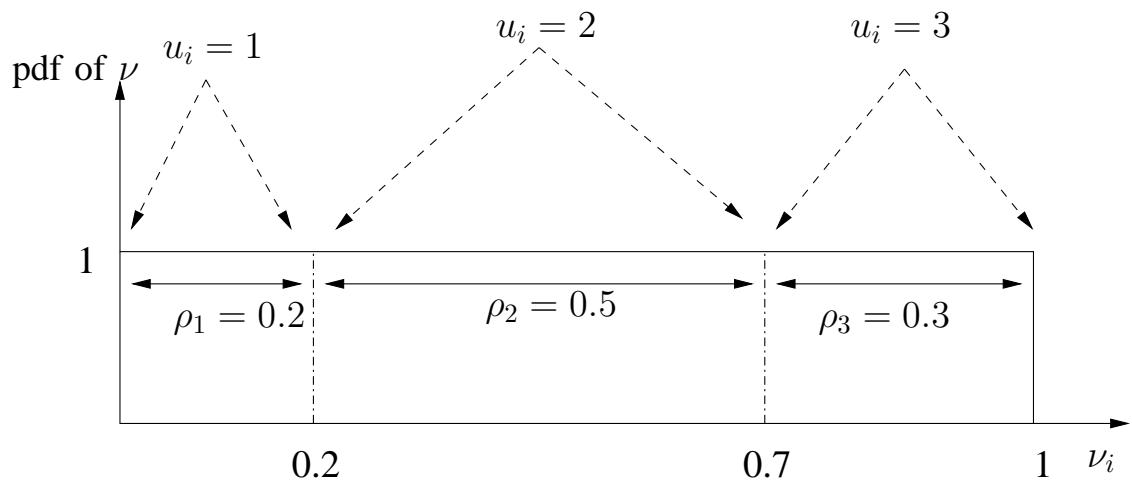}
\caption{Illustration of Proportional Expansion for discrete metrics. An example shown is for the case where the metric takes 3 values $1$, $2$, and $3$ with probabilities 0.2, 0.5, and 0.3, respectively.}
\label{fig:discrete metric}
\end{figure}

{\em Proportional Expansion}: Let the metric $u_i$ be a realization of an
$\omega$-valued discrete random variable that, without loss of
generality, takes values $1,2,\ldots,\omega$ with probability
$\rho_1,\rho_2,\ldots,\rho_{\omega}$, respectively. Each node
independently maps its metric $u_i$ into a new metric $\nu_i$ as
follows: {\em When $u_i = j$, $\nu_i$ is a realization of a uniformly
  distributed random variable in $\left(\sum_{\ell=0}^{j-1}\rho_\ell,
    \sum_{\ell=1}^{j}\rho_\ell \right)$, where $\rho_0 \define 0$.}
In other words, each node chooses a new random metric $\nu_i$ that is
uniformly distributed over a bin of length {\it proportional} to the
probability mass of its original metric $u_i$.

The overall distribution of the new metric across all users is then
uniformly distributed in $(0,1)$.  Proportional Expansion satisfies two key
properties:
\begin{itemize}
\item It preserves the sorting order of the metrics: if $u_{i} >
  u_{j}$, then $\nu_{i} > \nu_{j}$. Hence, selecting the best $Q$
  nodes with the highest $\nu_i$s is equivalent to selecting $Q$ nodes
  with the highest $u_i$s.

\item The probability that $\nu_i = \nu_j$, for $i \neq j$, is 0 since
  $\nu_i$ is a continuous random variable.
\end{itemize}

Therefore, the selection algorithm of Sec.~\ref{sec:Q-Relay Selection
  Algorithm} for any $Q \geq 1$ can then be run on $\nu_{i}$.  The
following Proposition formally quantifies the performance of
Proportional Expansion. It implies that proportional expansion is
scalable, \ie, it takes at most 2.47 slots for best relay selection,
$4.406$ slots for selecting the best 2 relays, and so on, for any
number of relays, $n$.

\begin{proposition}
\label{disc_metric}
The average number of slots required to select the best $Q$ relays by
Proportional Expansion for the discrete metrics case is the same as
that of the best $Q$ relay threshold based selection algorithm of
Sec.~\ref{subsec:Q-algo defiition} that operates on continuous
metrics.
\end{proposition}
\begin{proof}
The proof is omitted since it directly follows from the above discussion.
\end{proof}
%


%

\section{Conclusions}
\label{sec:conclusions}

We developed a new asymptotic analysis for the single relay splitting
based selection algorithm, which was based on a new Poisson point
process interpretation of the dynamics of the algorithm. This led to a
characterization of the optimal parameters of the algorithm, and
enabled a rigorous benchmarking of the greedy parameter setting used
in the literature. We also proposed a new splitting based algorithm
for selecting the best $Q$ relays, which are useful for several
cooperative protocols proposed in the literature. The new algorithm
was more efficient than running the single relay selection algorithm
multiple times. Furthermore, we generalized the analytical techniques
to handle multiple relay selection, and derived the exact expressions
for the average number of slots for multiple relay selection.
Interestingly, the asymptotic expressions were accurate even for a
small number of relays. With the help of proportional expansion, we
showed, for the first time, that splitting algorithms can be adapted
to work for discrete metrics as well without any loss in performance
or scalability whatsoever.

The analysis shows that the greedy policy of maximizing the success
probability in the next slot is suboptimal. While it works well for
single relay selection, it becomes more and more suboptimal as the
number of relays to be selected increases.  The analysis also shows
that the general single relay selection algorithm, the proposed
multiple relay selection algorithm, and the FCFS multiple access
control algorithm are intimately related. For example, the optimal
value of the contention load parameter increases as the number of
relays to be selected increases and finally approaches the optimal
setting for FCFS.  This is despite the fact that selection and
multiple access control algorithms serve very different purposes, and,
therefore, evaluated differently. While multiple access control
algorithms attempt to serve all nodes and are evaluated, for example,
by the maximum traffic they can handle with a finite delay, selection
algorithms are evaluated by how fast they can select the best nodes.
We hope that this insight will help develop better selection
algorithms.  An important property about splitting algorithms is that
besides being distributed, they are both extremely fast and scalable.
This suggests that selection based protocols will deliver improvements
in the overall end-to-end system-level performance even when the time
overhead incurred by the selection algorithm is accounted for. The
system-level benefits can be further improved if the multiple relay
selection algorithm proposed in this paper can be modified to allow
the metrics to be updated during the selection process.

\appendix

\subsection{Proof of Lemma~\ref{average_no_of_slots_finite_n}}
\label{proof of average_no_of_slots_finite_n}
It can be easily seen that the idle phase consists of at the most $q =
\ceil{\frac{n}{p_e}}-1$ slots since at this stage the lower threshold
equals the smallest value $0$.  Given that the first non-idle slot is
the $i\kth$ slot and $k$ nodes are involved, the average number of
slots required to find the best node is $\EX{k}+i$. (The recursive
expression for $\EX{k}$ is given in~ \cite[(6)]{qin_infocomm_2004}.)
The probability that the first non-idle slot is the $i\kth$ slot and
$k$ nodes transmit in it equals $
\binom{n}{k}\left(\frac{p_e}{n}\right)^{k}
\left(1-\frac{ip_e}{n}\right)^{n-k}$, for $i \le q$. This constitutes
the first term of the right side of~\eqn{eq:finite_nodes}.  The
probability that the $(q+1)\kth$ slot is the first non-idle slot is
$(1-\frac{qp_e}{n})^n$ since all nodes' metrics must lie in interval
$((q+1)p_{e},1]$. In the event that this happens, all $n$ nodes will
transmit and collide, which will take $\EX{n}$ slots to resolve.
Hence, the second term on the right side of~\eqn{eq:finite_nodes}
follows.

\subsection{Proof of Non-Recursive Expression of Theorem~\ref{m1}}
\label{proof of m2}
  Let the random variable $I$ denote the number of slots required until (and including) the first non-idle slot and $Y$ denote the number of slots required after that.

  Consider the state transition diagram of Figure~\ref{fig:markov
    chain}, in which the state represents the number of slots that
  have elapsed since the first non-idle slot. The node goes to state
  $\success$ whenever success occurs, and the algorithm terminates.
  Otherwise, in case of an idle or collision, the node increments its
  state by 1. By definition, state~$0$ is the first non-idle slot
  itself; thus, an idle outcome cannot occur in it. The following
  lemma is crucial in analyzing this transition diagram.
\begin{lemma}
The state transition diagram of Fig.~\ref{fig:markov chain} is a Markov chain.
\label{lem:markov}
\end{lemma}
\begin{proof}
To prove this, it is sufficient to prove that the
  transition probability from any state $i$ to $\success$ is dependent
  only on $i$. (Having done so, we shall denote this probability by $\PI{i}$.)

  We refer to the interval in $M(t)$ allocated to state $i$ as its
  threshold interval.  Here, $P_{0}$ is the probability that in a
  threshold interval of size $p_e$ only one node transmits given that
  at least one node transmits in that slot. Let $N(x)=M(t+x)-M(t)$.
  Then, from the memoryless property of the Poisson process,
  $\prob{N(x)=i}$ is independent of $t$ and is equal to $\frac{x^i
    e^{-x}}{i!}$. Thus,
 \begin{align}
 P_0 =\prob{N(p_e)=1 \Big| N(p_e) > 1} = \frac{p_e e^{-p_e}}{1-e^{-p_e}}.
 \end{align}
 $\PI{1}$ is the probability that the second non-idle slot is a
 success given that the first non-idle slot (of threshold interval
 size $p_e$) is a collision. Due to splitting, the second slot will
 have a threshold interval size that is half that of the first one.
 Therefore, $\PI{1}$ is the probability that conditioned on $N(p_e)$
 having at least 2 nodes (\ie, a collision), $N(p_e/2)$ has exactly
 one. Thus,
  \begin{equation}
  \PI{1}=\prob{N\left(\frac{p_e}{2}\right)=1 \Big| N(p_e) \geq 2}
 =\frac{\prob{N\left(\frac{p_e}{2}\right)=1 , N(p_e) \geq 2}}{Pr(N(p_e) \geq 2)}.
\end{equation}
Therefore,
\begin{equation}
            \PI{1} = \frac{\prob{N\left(\frac{p_e}{2}\right)=1 , N(p_e)- N\left(\frac{p_e}{2}\right)\ge1}}{Pr(N(p_e) > 1)}
             = \frac{ \frac{p_e}{2} e^{-\frac{p_e}{2} }(1-e^{-\frac{p_e}{2} })}{1-(1+p_e )e^{-p_e }}.
  \end{equation}

  For $\PI{2}$, the following two trajectories can occur: State 2 was
  reached by a collision in state 1 or by an idle in state 1. In case
  of a collision, the threshold interval of the second non-idle slot
  (of size $p_e/2$) gets split into two halves. Even in the case of an
  idle the interval would be split into two halves and nodes from the
  left half would contend.  Thus, $\PI{2}$ is equal to the probability
  that conditioned on an interval of size $p_e/2$ having at least two
  nodes, half the interval (of size $p_e/4$), has exactly one node. Thus,
  \begin{equation*}
  \PI{2}= \prob{N\left(\frac{p_e}{4}\right)=1 \Big| N\left(\frac{p_e}{2}\right) \geq 2}
             = \frac{\prob{N\left(\frac{p_e}{4}\right)=1 , N\left(\frac{p_e}{2}\right) \geq 2}}{Pr(N(p_e) \geq 2)} = \frac{\frac{p_e}{4}
  e^{-\frac{p_e}{4} }(1-e^{-\frac{p_e}{4} })}{1-(1+\frac{p_e}{2}
  )e^{-\frac{p_e}{2} }}.
\end{equation*}

In similar way, we can show that $\PI{i},~\forall~i\ge 1$, is equal to
the probability that conditioned on an interval of size
$2^{-(i-1)}p_e$ having at least two nodes, one half of the interval (of size
$2^{-i}p_e$) has exactly one node. Thus,
    \begin{equation}
  \PI{i}= \prob{N\left(\frac{p_e}{2^{i}}\right)=1 \Big| N\left(\frac{p_e}{2^{i-1}}\right) > 1}
             = \frac{2^{-i}p_ee^{-2^{-i}p_e}(1-e^{-2^{-i}p_e})}{1-(1+2^{-(i-1)}p_e)e^{-2^{-(i-1)}p_e}}.
  \end{equation}
\end{proof}

Now, $\aveslots{\infty}(p_e)=\expect{I}+\expect{Y}$. From the Poisson
process interpretation of Theorem~\ref{m1}, we can show that
$\prob{I=i} = e^{-(i-1)p_e}(1-e^{-p_e})$. Therefore, $\expect{I} =
\sum_{i=1}^\infty ie^{-(i-1)p_e}(1-e^{-p_e}) = \frac{1}{1-e^{-p_e}}$.
The average number of slots required after the first non-idle slot to
select the best node, $\expect{Y}$, is calculated as follows.  First,
\mbox{$\expect{Y} = \sum_{i=1}^\infty i \Pr(Y=i)$}, can be shown to be
identically equal to $\sum_{i=1}^\infty \Pr(Y \geq i)$.  Second, since
each state in the Markov chain is visited at most once, it follows
that $\expect{Y} = \sum_{i=1}^\infty p(i)$, where $p(i)$ is the
probability that the $i\kth$ state is visited.  From the state
transition diagram, it is easy to see that \mbox{$p(i)=
  (1-\PI{0})\prod_{j=1}^{i-1} (1-\PI{j})$ }. Hence, the desired
expression for $\aveslots{\infty}(p_e)$ follows.

\begin{figure}[p]
\centering
\input{markov_single.pstex_t}
\caption{State transition diagram for the number of slots required to select the best node after the first non-idle slot.}
\label{fig:markov chain}
\end{figure}

\subsection{Proof of Corollary~\ref{upperbound}}
\label{proof of upperbound}

From~\cite{qin_infocomm_2004}, we have $\EX{k} \le \log_2(k)+1$, $k
\geq 2$, and $\EX{1}=0$. Since $\log_2(x)$ is concave with respect to
$x$, a tangent to it at any point $(k_0,\log_{2}(k_0))$ is an upper
bound. Therefore, 
\begin{equation}
\log_2(k) \le \frac{k-k_0}{k_0\log_{e}(2)} +
\log_2(k_0).  
\end{equation}
Consequently, $\EX{k} \le \frac{k}{k_{0} \log_{e}(2)} + \log_{2}(2
k_{0}/e)$, $k \geq 2$. Substituting this in~\eqn{eq:m1}, we get
\begin{align}
 \aveslots{\infty}(p_e)  \leq \log_{2}\brac{\frac{2 k_{0}}{e}} \frac{\sum_{k=2}^{\infty} \frac{p_{e}^{k}}{k!}}{e^{p_{e}} - 1} + \frac{1}{k_{0} (e^{p_{e}} - 1) \log_{e}(2)} \sum_{k=2}^{\infty} k \frac{p_{e}^{k}}{k!} + \frac{1}{1 - e^{-p_{e}}}.
\label{eq:ub}
\end{align}

For $k_0 \geq e/2$, $\log_{2}(2 k_{0}/e) \geq 0$. Also,
$\sum_{k=2}^{\infty} \frac{p_{e}^{k}}{k!} = e^{p_e}-1-p_e < e^{p_e}-1$
since $p_{e} > 0$. Therefore, for $k_0 \geq e/2$, the first term in
the right hand side of~\eqn{eq:ub} is less than $\log_{2}\brac{\frac{2
    k_{0}}{e}}$. Substituting this inequality in~\eqn{eq:ub} and
simplifying leads to the desired result in~\eqn{eq:upperbound}.


\subsection{Proof of Theorem~\ref{2 best nodes}}
\label{proof of 2 best nodes}

{\it Proof of~\eqn{aveslots2}}: Given that the first non-idle slot is
the $i\kth$ slot and $k\ge1$ nodes are involved, the average number of
slots required to select the best $2$ nodes is $\EXQ{k}{2}+i$.  The
probability that the first non-idle slot is the $i\kth$ slot and
$k\ge1$ nodes are involved is $e^{-ip_e} p_e^k/k!$.  Hence, we get
\begin{equation}
\aveslotsQ{\infty}{2}(p_e) = \sum_{i=1}^\infty\sum_{k=1}^\infty e^{-ip_e}\frac{p_e^k}{k!} \left(\EXQ{k}{2}+i\right),
\end{equation}
simplifying which yields~\eqref{aveslotsQ}.

If only one node transmits in the first non-idle slot, then a success
occurs and the node gets selected. Selecting one more node will take
$\aveslotsQ{\infty}{1}(p_e)$ slots, on average. (This follows from the
memoryless property of the Poisson process~\cite{wolff}.) Thus,
$\EXQ{1}{1}=\aveslotsQ{\infty}{1}(p_e)$.  Also, if exactly two nodes
transmit in the first non-idle slot, only one node transmits in the
slot just after the first success. Thus, $\EXQ{2}{2}=\EXQ{1}{2}+1=3$
slots.  When $k>3$ nodes transmit in the first non-idle slot, the
following three cases are possible for the next slot: (i)~{\it
  Collision among $i$ nodes}: $\EXQ{i}{2}$ more slots would then be
required, on average.  (ii)~{\it Idle}: $\EXQ{k}{2}$ more slots are
required, on average. (iii)~{\it Success}: The next slot would then
surely involve a collision among $k-1$ nodes.  $\EXQ{k-1}{1}$ slots,
on average, would be required {\it after} that.  The probability that
$i$ nodes transmit in the next slot is $\binom{k}{i}/2^{k}$. Thus,
\begin{multline} \label{collision 2 users}
\EXQ{k}{2}=\frac{1}{2^{k}}\Bigg(\left(\binom{k}{0}+\binom{k}{k}\right)\left(1+\EXQ{k}{2}\right) \\
+ \binom{k}{1}\left(1+\EXQ{k-1}{1}+1\right) +  \sum_{i=2}^{k-1}\!\! \binom{k}{i}\left(1+\EXQ{i}{2}\right)\Bigg).
\end{multline}
Simplifying this further using combinatorial
identities~\cite{gradshteyn00_book} results in~\eqref{EX_k[2]}.

{\it Proof of~\eqn{aveslots2_2}:} This proof also involves
constructing a state transition diagram that will be proved to be a
Markov chain. Consider the state transition diagram of
Figure~\ref{fig:Markov2Users}. It is more involved than that in
Figure~\ref{fig:markov chain} because we need to also track how many
successes have occurred.  State $i$ corresponds to the $i\kth$ split
before the first success (which takes $i$ slots), state $i'$
corresponds to the first success occurring at the $i$th slot, and
state $i''$ corresponds to the first success having already occurred
by the $i$th slot. The state transition diagram can be explained in
detail as follows.

State~$0$ corresponds to the first non-idle slot. If the first
non-idle slot is a success, the node moves from state~$0$ to
state~$\success_1$. Now, the algorithm starts a new collision
resolution to find the second colliding node. This takes time
$\aveslotsQ{\infty}{1}(p_e)$, which is given by Theorem~\ref{m1}. If
the first non-idle slot is a collision, its threshold interval is
split and the node transitions from state~$0$ to state~$1$. Each
subsequent idle or collision results in one additional split and the
node moves from state~$i$ to~$i+1$.  In case of a success, the node
moves from from state~$i$ to state~$i'$ as no additional split occurs.
A success in state~$i'$ results in a transition to state~$\success$,
at which time the algorithm terminates.  In case of a collision in
state~$i'$, the node moves to state~$(i+1)''$, as one more split
occurs.  In case of a success in state $i''$, the node moves to state
$\success$, and the algorithm terminates. Otherwise, an idle or
collision results in a transition from state~$i''$ to $(i+1)''$.  Note
that in each state ($i$, $i'$, or $i''$) the size of threshold
interval is $2^{-i}p_e$.

The following Lemma shall prove to be crucial in analyzing this
transition diagram.
\begin{lemma}
The state transition diagram of Fig.~\ref{fig:Markov2Users} is a Markov chain.
\label{lem:markov2users}
\end{lemma}
\begin{proof}
For this, it is sufficient to prove that transition probabilities for each state depend only on the $i$ and not on
the path taken to reach that state.

Let $ P_{0}$ be the probability of success in state~$0$ (the first
non-idle slot). It is equal to the probability that in a slot of size
$p_e$ only one node transmits given that at least one node transmits
in that slot. Let $N(x)=M(t+x)-M(t)$. Then, by memoryless property of
the Poisson process, $\prob{N(x)=i}$ is independent of $t$ and is equal to
$\frac{x^i e^{-x}}{i!}$. Thus,
 \begin{equation}
 P_0= \prob{N(p_e)=1 \Big| N(p_e) > 1}
    =\frac{p_e e^{-p_e}}{1-e^{-p_e}}.
 \end{equation}
Let $\PLi{i}$ be the probability of success in state $i$, which is equal to the probability
that given an interval of size $2^{-i}p_e$
having more than one nodes, left half of it has exactly one node.  Thus,
 \begin{equation}\label{eq:PLi}
 \PLi{i}= \prob{N\left(\frac{p_e}{2^{i}}\right)=1 \Big| N\left(\frac{p_e}{2^{i-1}}\right) > 1}
             = \frac{2^{-i}p_ee^{-2^{-i}p_e}(1-e^{-2^{-i}p_e})}{1-(1+2^{-(i-1)}p_e)e^{-2^{-(i-1)}p_e}}.
 \end{equation}
 Let $\PRi{i}$ be the probability of success in state $i'$. State~$i'$
 can be entered only after success in state~$i$. Thus, threshold
 interval of state~$i$, which is right half of the split during
 state~$i-1$, has {\it at least} one node. Thus $\PRi{i}$ is equal to
 the probability that exactly one node transmits in the slot with
 interval size $2^{-i}p_e$, given that at least one node lies in that
 interval, which equals
 \begin{equation}
 \PRi{i}= \prob{N\left(\frac{p_e}{2^{i}}\right)=1 \Big| N\left(\frac{p_e}{2^{i}}\right) > 1}
             = \frac{2^{-i}p_ee^{-2^{-i}p_e}}{1-e^{-2^{-i}p_e}}.
\label{eq:PRi}
 \end{equation}

 The probability of success in state $i''$ is again equal to the
 probability that given that an interval of size $2^{-i}p_e$ has more
 than one node, its left half has exactly one node. This probability
 equals $\PLi{i}$. Thus, from~\eqn{eq:PLi} and \eqn{eq:PRi}, the
 transition probabilities $\PLi{i}$ and $\PRi{i}$ only depend on $i$,
 which proves that Fig.~\ref{fig:Markov2Users} is a Markov chain.
\end{proof}

 Let the random variable $I$ denote the number of slots required until (and including) the first non-idle slot and $Y$ denote the number of slots required after that. Then, $\aveslotsQ{\infty}{2}(p_e)=\expect{I}+\expect{Y}$. Again, using Poisson point process interpretation,
   $\prob{I=i} = e^{-(i-1)p_e}(1-e^{-p_e})$, which implies,
\begin{equation}
\expect{I} = \sum_{i=1}^\infty ie^{-(i-1)p_e}(1-e^{-p_e}) =
  \frac{1}{1-e^{-p_e}}. 
\end{equation}

$\expect{Y}$ can be calculated from Lemma~\ref{lem:markov2users} as
follows.  Let $p(i)$, $p'(i)$, and $p''(i)$ be the probability that
states~$i$, $i'$, and $i''$ are visited, respectively. From the state
transition diagram, since state~$i$ can be reached only from
state~$i-1$ and state~$i'$ can be reached only from state~$i$, we get
$ p(i)= (1-P_0)\prod_{j=1}^{i-1} (1-\PLi{j}),\forall i\ge1,$ and $
p'(i)=p(i)\PLi{i} ,\forall i\ge1,$.  Also, since for $i>2$ state~$i''$
can be reached from $(i-1)'$ and $(i-1)''$, we get
\begin{equation}
p''(i) = p'(i-1)(1-\PRi{i-1}) + p''(i-1)(1-\PLi{i-1}),\quad\forall~i > 2,
\end{equation}
and  
\begin{equation}
p''(2)=p'(1)(1-\PRi{i-1}).
\end{equation}

Now, if state $\success_1$ is visited $\aveslotsQ{\infty}{1}(p_e)$
slots, on average, are required, which occurs with probability $P_0$.
Else, the average number of slots is equal to
$\sum_{j=1}^{\infty}\prob{Z\ge j}$, where $Z$ is the total number of
states visited excluding state~$0$. This is so because we are counting
the number of slots required {\it after} the first non-idle slot.
Since each state is visited at most once, the average above is equal
to $\sum_{i=1}^{\infty}\left(p(i)+p'(i)+p''(i)\right)$.  Thus, the
average number of slots required after the first non-idle slot is
$\expect{Y} = P_0
\aveslotsQ{\infty}{1}(p_e)+\sum_{i=1}^{\infty}\left(p(i)+p'(i)+p''(i)\right)$.

\begin{figure}[p]
\centering
\input{Markov2Users.pstex_t}
\caption{ State transition diagram for the number of slots required to select the best two nodes after the first non-idle slot.}
\label{fig:Markov2Users}
\end{figure}

\subsection{Proof of Theorem~\ref{Q best nodes}}
\label{proof of Q best nodes}
The proof is similar to the proof of~\eqn{aveslots2} in Theorem~\ref{2
  best nodes}, except for the following differences:
\begin{enumerate}

\item When two nodes transmit in the first non-idle slots,
  $\EXQ{2}{2}=3$ slots, on average, are required to select both of
  them.  Selecting the remaining best $Q-2$ nodes takes another
  $\aveslotsQ{\infty}{Q-2}(p_e)$ slots, on average.  Thus,
  $\EXQ{2}{Q}=\EXQ{2}{2}+\aveslotsQ{\infty}{Q-2}(p_e)$.

\item When $k>3$ nodes transmit in the first non-idle slot, the
  average number of slots required thereafter is
\begin{multline}
\EXQ{k}{Q}=0.5^k \Bigg[\left(\binom{k}{0}+\binom{k}{k}\right)\left(1+\EXQ{k}{Q}\right) \\
+ \binom{k}{1}\left(1+\EXQ{k-1}{Q-1}+1\right) +  \sum_{i=2}^{k-1} \binom{k}{i}\left(1+\EXQ{i}{Q}\right)\Bigg].
\end{multline}
\end{enumerate}

\bibliographystyle{ieeetr}
\bibliography{../../../Bibtex/bibJournalList,../../../Bibtex/cooperativeComm,../../../Bibtex/database,../../../Bibtex/lognormal,../../../Bibtex/MIMO,../../../Bibtex/mimoEstimation,../../../Bibtex/cdma,../../../Bibtex/adaptation,../../../Bibtex/scheduling,../../../Bibtex/standard,../../../Bibtex/book}

\end{document}